\documentclass[11pt]{article}
\usepackage{amsmath,amsthm,amssymb,hyperref,enumerate,xcolor, upgreek,bbm}
\usepackage[ruled]{algorithm2e}
\usepackage{framed,empheq,url}
\usepackage{mathtools}
\usepackage[symbol]{footmisc}

\newcommand{\remove}[1]{}
\newtheorem{thm}{Theorem}[section]  
\newtheorem{claim}[thm]{Claim}
\newtheorem{lem}[thm]{Lemma}
\newtheorem{define}[thm]{Definition}
\newtheorem{cor}[thm]{Corollary}

\newtheorem{THM}{Theorem}

\newtheorem{remark}[thm]{Remark}

\DeclarePairedDelimiter\ceil{\lceil}{\rceil}

\newcommand{\bigo}[1]{O\left(#1\right)}
\newcommand{\set}[1]{\{#1\}}
\newcommand{\BF}{\set{0,1}}
\newcommand{\adj}{\mathrm{adj}}

\newcommand{\vect}[2]{#1_1,\cdots,#1_#2}
\newcommand{\bigexp}[1]{\exp\left(#1\right)}
\newcommand{\twovect}[2]{\left[\begin{matrix}
#1\\
#2\\
\end{matrix}\right]}

\newcommand{\fourmatrix}[4]{\left[\begin{matrix}
#1& #2\\
#3& #4\\
\end{matrix}\right]}

\def\tb{\textbf}

\def\F{{\mathbb{F}}}

\def\Z{{\mathbb{Z}}}

\def\cP{{\cal P}}

\def\cS{{\mathcal S}}

\def\cF{\mathcal F}

\def\cU{\mathcal U}

\def\cR{\mathcal{R}}

\def\cQ{\mathcal{Q}}
\def\cA{\mathcal{A}}
\def\cP{\mathcal{P}}
\def\cV{\mathcal{V}}

\def\ba{{\mathbf a}}

\def\bc{{\mathbf c}}

\def\bu{{\mathbf u}}
\def\bv{{\mathbf v}}
\def\bx{{\mathbf x}}

\def\bz{{\mathbf z}}
\def\bgam{\gamma}
\def \blam{{\boldsymbol\lambda}}

\newcommand{\inpro}[2]{\langle #1,#2 \rangle}

\def\_{\,\,\,\,\,}

\newcommand{\ignore}[1]{}

\allowdisplaybreaks

\begin{document}

\title{2-Server PIR with sub-polynomial communication}

\author{Zeev Dvir\thanks{Department of Computer Science and Department of Mathematics,
Princeton University.
Email: \texttt{zeev.dvir@gmail.com}. Research  supported by NSF grants CCF-1217416
and CCF-0832797.} \and
Sivakanth Gopi\thanks{Department of Computer Science, Princeton University.
Email: \texttt{sgopi@cs.princeton.edu}.}}

\date{}
\maketitle

\thispagestyle{empty}
\pagenumbering{arabic}

\begin{abstract}
A 2-server Private Information Retrieval (PIR) scheme allows a user to retrieve the $i$th bit of an $n$-bit database replicated among two servers (which do not communicate) while not revealing any information about $i$ to either server. In this work we construct a 1-round 2-server PIR with total communication cost $n^{O({\sqrt{\log\log n/\log n}})}$. This improves over the currently known 2-server protocols which require $O(n^{1/3})$ communication  and matches the communication cost of known $3$-server PIR schemes.  Our improvement comes from reducing the number of servers in  existing  protocols, based on Matching Vector Codes,  from 3 or 4 servers to 2. This is achieved by viewing these protocols in an algebraic way (using polynomial interpolation) and extending them using partial derivatives. 
\end{abstract}

\thispagestyle{empty}

\pagenumbering{arabic}

\section{Introduction}
Private Information Retrieval (PIR) was first introduced  by Chor, Goldreich, Kuzhelevtiz and Sudan \cite{ChorKGS98}. In a $k$-server PIR scheme, a user can retrieve the $i$th bit $a_i$ of a $n$-bit database $\ba=(\vect{a}{n})\in\BF^n$ replicated among $k$ servers (which do not communicate) while giving no information about $i$ to any server. The goal is to design PIR schemes that minimize the communication cost which is the worst case number of bits transferred between the user and the servers in the protocol. The trivial solution which works even with one server is to ask a server to send the entire database $\ba$, which has communication cost $n$. 

When $k=1$ the trivial solution cannot be improved \cite{ChorKGS98}. But when $k\ge 2$, the communication cost can be brought down significantly. In \cite{ChorKGS98}, a 2-server PIR scheme with communication cost $O(n^{1/3})$ and a $k$-server PIR scheme with cost $\bigo{k^2\log k n^{1/k}}$ were presented. The $k$-server PIR schemes were improved further in subsequent papers \cite{Ambainis97,BeimelI01,BeimelIKR02}. In \cite{BeimelIKR02}, a $k$-server PIR scheme with cost $n^{\bigo{\frac{\log\log k}{k\log k}}}$ was obtained. This was the best for a long time until the breakthrough result of Yekhanin\cite{Yekhanin08} who gave the first $3$-server scheme with sub-polynomial communication (assuming a number theoretic conjecture). Later, Efremenko\cite{Efremenko09} gave an unconditional $k$-server PIR scheme with sub-polynomial cost for $k\ge 3$ which were slightly improved in \cite{ItohS10} and \cite{CheeFLWZ13}. These new PIR schemes follow from the constructions of constant query smooth Locally Decodable Codes (LDCs) of sub-exponential length called Matching Vector Codes (MVCs)\cite{DvirGY10}. A $k$-query LDC \cite{KT00} is an error correcting code which allows the receiver of a corrupted encoding of a message to recover the $i$th bit of the message using only $k$ (random) queries. In a {\em smooth} LDC, each query of the reconstruction algorithm is uniformly distributed among the code word symbols. Given a $k$-query smooth LDC, one can construct a $k$-server PIR scheme by letting each server simulate one of the queries. Despite the advances in $3$-server PIR schemes, the 2-server PIR case is still stuck at $O(n^{1/3})$ since 2-query LDCs provably require exponential size encoding \cite{KerenidisW03} (which translates to $\Omega(n)$ communication cost in the corresponding PIR scheme). For more information on the relation between PIR and LDC and the constructions of sub-exponential LDCs and sub-polynomial cost PIR schemes with more than 2 servers we refer to the survey \cite{Yekhanin12}.  

On the lower bounds side, there is very little known. The best known lower bound for the communication cost of a 2-server PIR is $5\log n$ \cite{WehnerW05} whereas the trivial lower bound is $\log n$. In \cite{ChorKGS98}, a lower bound of $\Omega(n^{1/3})$ is conjectured. In \cite{RazborovY06}, an $\Omega(n^{1/3})$ lower bound was proved for a restricted model of 2-server PIR called bilinear group based PIR. This model encompasses all the previously known constructions which achieve $O(n^{1/3})$ cost for 2-server PIR. We elaborate more on the relation between this model and our construction after we present our results below.

PIR is extensively studied and there are several variants of PIR in literature. The most important variant with cryptographic applications is called Computationally Private Information Retrieval (CPIR). In CPIR, the privacy guarantee is based on computational hardness of certain functions i.e. a computationally bounded server cannot gain any information about the user's query. In this case, non-trivial schemes exist even in the case of one server under some cryptographic hardness assumptions. For more information on these variants of PIR see \cite{Gasarch,Gasarch04,Lipmaa}. In this paper, we are only concerned with information theoretic privacy i.e. even a computationally unbounded server cannot gain any information about the user's query which is the strongest form of privacy.

\subsection{Our Results}

We start with a formal definition of a 2-server PIR scheme. A 2-server PIR scheme involves two servers $\cS_1$ and $\cS_2$ and a user $\cU$. A database $\ba=(\vect{a}{n})\in \BF^n$ is replicated between the servers $\cS_1$ and $\cS_2$. We assume that the servers cannot communicate with each other. The user $\cU$ wants to retrieve the $i$th bit of the database $a_i$ without revealing any information about $i$ to either server. The following definition is from \cite{ChorKGS98}:

\begin{define}

A 2-server PIR protocol is a triplet of algorithms $\cP=(\cQ,\cA,\cR)$. At the beginning, the user $\cU$ obtains a random string $r$. Next $\cU$ invokes $\cQ(i,r)$ to generate a pair of queries $(q_1,q_2)$. $\cU$ sends $q_1$ to $\cS_1$ and $q_2$ to $\cS_2$. Each server $S_j$ responds with an answer $ans_j=\cA(j,\ba,q_j)$. Finally, $\cU$ computes its output by applying the recovery algorithm $\cR(ans_1,ans_2,i,r)$. The protocol should satisfy the following conditions:
\begin{itemize}
\item \tb{Correctness : } For any n, $\ba\in \BF^n$ and $i\in [n]$, the user the outputs the correct value of $a_i$ with probability 1 (where the probability is over the random strings $r$) i.e. $\cR(ans_1,ans_2,i,r)=a_i$
\item \tb{Privacy : } Each server individually learns no information about $i$ i.e.   for any fixed database $\ba$ and for $ j=1,2$, the distributions of $q_j(i_1,r)$ and $q_j(i_2,r)$ are identical for all $i_1,i_2\in [n]$ when $r$ is randomly chosen.
\end{itemize}
The communication cost of the protocol is the total number of bits exchanged between the user and the servers in the worst case.
\end{define}

$k$-server PIR is similarly defined, with the database replicated among $k$ servers which cannot communicate between themselves. We only defined 1-round PIR i.e. there is only one round of interaction between the user and the servers. All known constructions of PIR schemes are 1-round and it is an interesting open problem to find if interaction helps. We now state our main theorem:

\begin{THM}
\label{mainthm}
There exists a 2-server PIR scheme with communication cost 
 $n^{\bigo{\sqrt{\frac{\log\log n}{\log n}}}}$.
\end{THM}

The  definition of a 2-server PIR scheme can be generalized in an obvious manner to any number of servers. In \cite{Efremenko09} a  $2^r$-server PIR schemes was given with $n^{\bigo{({\log\log n}/{\log n})^{1-1/r}}}$ communication cost for any $r\ge 2$. Using our techniques, we can reduce the number of servers in this scheme by a factor of two. That is, we prove the following stronger form of Theorem~\ref{mainthm}.

\begin{THM}\label{THM-kserver}
For any $r \geq 2$,	there exists a $2^{r-1}$-server PIR scheme with communication cost $n^{\bigo{({\log\log n}/{\log n})^{1-1/r}}}$. 
\end{THM}

We note that the proof of Theorem~\ref{THM-kserver} actually allows the database symbols to be in the larger alphabet $\Z_m$, where $m$ is the composite over which we construct the MV family.

There was some work on decreasing the $2^r$ query complexity of the construction of Matching Vector Codes in \cite{Efremenko09}. A query complexity of $9\cdot 2^{r-4}$ for $r\ge 6$ was achieved in \cite{ItohS10} while keeping the encoding length the same. This was improved in \cite{CheeFLWZ13} to $3^{\ceil{r/2}}$ for $2\le r\le 103$ and $(\frac{3}{4})^{51}\cdot 2^r$ for $r\ge 104$. Using these LDCs directly to get a PIR scheme will do better than our scheme when the number of servers is more than 26, whereas our scheme will do better than these when the number of servers are less than 9.

\subsection{Related work}

\paragraph{Polynomial lower bounds for bilinear group based PIR: } In \cite{RazborovY06}, an $\Omega(n^{1/3})$ lower bound was shown  for a restricted model of  2-server PIR schemes. This lower bound holds for schemes that are both {\em bilinear} and {\em group based}. Our scheme can be made into a bilinear scheme\footnote{Our scheme can infact be made linear and using a simple transformation given in \cite{RazborovY06}, any linear scheme can be converted to a bilinear scheme} (see section \ref{sec-overZm}) over the field $\F_3$ of three elements. However, it does not satisfy the property of being group based as defined in \cite{RazborovY06}. Our scheme does satisfy a weaker notion of {\em employing a group-based secret sharing scheme} (another technical term defined in \cite{RazborovY06}). The difference between these two notions (of being group based as opposed to employing a group based secret sharing scheme) is akin to the difference between LCCs and LDCs (LCCs being the stronger notion). 
In group based PIR, the database is represented by the values of a function over a subset of a group but the user should be able to recover the value of that function at {\em every} group element. Our scheme encodes the database as a function over a group and the user will only be able to recover the bits of the database from the function.

\paragraph{2-query LDCs over large alphabet:} The reader familiar with the exponential lower bounds for 2-query LDCs \cite{KerenidisW03} would wonder why our construction does not violate these bounds. The reason is that, when one translates 2-server PIR schemes into LDC, the resulting alphabet of the code can be quite large. Formally, a scheme with communication cost $s$ will translate into an LDC $C : \{0,1\}^n \mapsto (\{0,1\}^s)^{2^s}$ (with the blocks corresponding to all possible answers by the servers). Thus, each one of the two queries  used by the decoder is a string of $s$ bits. The known lower bounds for such LDCs are  exponential only as long as $s << \log(n)$ and so our construction does not violate them. Hence, our main theorem also gives the first construction of a sub-exponential 2-query LDC over an alphabet of size $2^{n^{o(1)}}$.

\subsection{Proof Overview}
On a very high level, the new protocol combines the existing 2-server scheme of \cite{WoodruffY05}, which uses polynomial interpolation using derivatives, with  Matching Vector Codes (MV Codes) \cite{Yekhanin08,Efremenko09}. In particular, we make use of the view of MV codes as polynomial codes, developed in \cite{DvirGY10}. This short overview is meant as a guide to the ideas in the construction and assumes some familiarity with \cite{WoodruffY05} and \cite{DvirGY10} (a detailed description  will follow in the next sections). The 2-server scheme of \cite{WoodruffY05} works by embedding the database $\ba = (a_1,\ldots,a_n)$ as evaluations of a degree 3 polynomial $F(x_1,\ldots,x_k)$ over a small finite field $\F_q$ with $k \sim n^{1/3}$. To recover the value $a_i = F(P_i)$ the user passes a random line through the point $P_i \in \F_q^k$, picks two random points $Q_1,Q_2$ on that line and sends the point $Q_i$ to the $i$th server. Each server responds with the value of $F$ at $Q_i$ and the values of all partial derivatives $\partial F / \partial x_j, j=1,\ldots,k$  at that point. The restriction of $F$ to the line is a univariate degree 3 polynomial and the user can recover the values of this polynomial at two points as well as the value of its derivative at these points. These four values (two evaluations plus two derivatives) are enough to recover the polynomial and so its value at $P_i$. The key is that each server's response only depends on the point $Q_i$ (which is completely random). The user can compute the derivatives of the restricted polynomial from these values (knowing the line equation).

To see how MV codes come into the picture we have to describe them in some detail. An MV family is a pair of lists $\cU = (\bu_1,\ldots,\bu_n)$, $\cV = (\bv_1,\ldots,\bv_n)$ with each list element $\bu_i$ and $\bv_j$ belonging to $\Z_m^k$ and $m$ is a small integer. These lists must satisfy the condition that $\inpro{\bu_i}{\bv_j}$ (taken mod $m$) is zero iff $i=j$. When $m$ is a composite, one can construct  such families of vectors of size $n = k^{\omega(1)}$ \cite{Grolmusz99} (this is impossible if $m$ is prime). From such a family we can construct an $m$-query LDC as follows: given a message $\ba = (a_1,\ldots,a_n) \in \{0,1\}^n$ define the polynomial $F(x_1,\ldots,x_k) = \sum_{i=1}^n a_i\bx^{\bu_i}$ (we denote $\bx^\bc = x_1^{c_1}\ldots x_k^{c_k}$). We think of $F$ as a polynomial with coefficients in some finite field $\F_q$ containing an element $\gamma \in \F_q$ of order $m$. The final encoding of $\ba$ is the evaluations of $F$ over all points in $\F_q^k$ of the form $\gamma^{\bc} = (\gamma^{c_1},\ldots,\gamma^{c_k})$ for all $\bc \in \Z_m^k$. To recover $a_i$ in a `smooth' way, we pick a random $\bz \in \Z_m^k$ and consider the restriction of $F$ to the `multiplicative line' given by $L = \{ \gamma^{\bz + t \bv_i}\,|\, t \in \Z_m\}$. That is, we denote $G(t) = F(\gamma^{\bz + t \bv_i})$. In \cite{DvirGY10} it was observed that this restriction can be seen as a polynomial $g(T)$ of degree at most $m-1$ in the new `variable' $T = \gamma^t$ and so can be reconstructed from the $m$ values on the line $g(\gamma^t) = G(t), t = 0,1,\ldots,m-1$. The final observation is that $g(0)$ is a nonzero multiple of $a_i$ (since the only contribution to the free coefficient comes from the monomial $a_i\bx^{\bu_i}$) and so we can recover it if we know $g(T)$.

Our new protocol combines these two constructions by using the MV code construction and then asking each server for the evaluations of $F$ at a point, as well as the values of a certain differential operator (similar to first order derivatives) at these points. For this to work we need two ingredients. The first is to replace the field $\F_q$ with a certain ring which has characteristic $m$ and an element of order $m$ (we only use $m=6$ and can take the polynomial ring $\Z_m[\gamma]/(\gamma^6 - 1)$). The second is an observation that, in known MV families constructions \cite{Grolmusz99}, the inner products $\inpro{\bu_i}{\bv_j}$ that are nonzero (that is, when $i \neq j$) can be made to fall in a small set. More precisely, over $\Z_6$, the inner products are either zero or in the set  $\{1,3,4\}$. This means that the restricted polynomial only has nonzero coefficients corresponding to powers of $T$ coming from the set $\{0,1,3,4\}$. Such a polynomial has four degrees of freedom and can be recovered from two evaluations and two derivatives (of order one). We are also able to work with arbitrary MV families by using second order derivatives at two points (which are sufficient to recover a degree 5 polynomial).

\subsection{Organization}
In section \ref{preliminaries} we give some preliminary definitions and notations. In section \ref{oldconstruction}, we review the construction of a 2-server PIR scheme with $O(n^{1/3})$ communication cost which is based on polynomial interpolation with partial derivatives \cite{WoodruffY05}. In section \ref{newconstruction}, we present our new construction of sub-polynomial 2-server PIR schemes and some of its variants. Then, in Section~\ref{sec-kserver} we analyze the generalization to more servers.  We conclude in Section~\ref{sec-conclude} with some remarks on future directions.

\section{Preliminaries}
\label{preliminaries}

\paragraph{Notations:}
We will use bold letters like $\bu,\bv,\bz$ etc. to denote vectors. The inner product between two vectors $\bu=(\vect{u}{k}),\bv=(\vect{v}{k})$ is denoted by $\inpro{\bu}{\bv}=\sum_{i=1}^k u_iv_i$. For a commutative ring $\cR$ we will denote by $\cR[\vect{x}{k}]$ the ring of polynomials in formal variables $x_1,\ldots,x_k$ with coefficients in $\cR$. We will use the notation $\bx^\bz$ with $\bx=(\vect{x}{k}),\ \bz=(\vect{z}{k}) \in \Z^k$ to denote the monomial $\prod_{i=1}^k x_i^{z_i}$. So any polynomial $F(\bx)\in \cR[\vect{x}{k}]$ can be written as $F(\bx)=\sum_{\bz} c_{\bz}\bx^{\bz}$. 

$\Z_m=\Z/m\Z$ is the ring of integers modulo $m$. When $\bu\in \Z_m^k$, $\bx^\bu$ denotes $\bx^{\tilde{\bu}}$ where $\tilde{\bu}\in \set{0,1,\cdots,m-1}^k$ is the unique vector such that $\bu\equiv \tilde{\bu} \mod m$. $\F_q$ denotes the finite field of size $q$.

\subsection{The rings $\cR_{m,r}$}


For our construction it will be convenient (although not absolutely necessary, see Section~\ref{sec-overZm}) to work over a ring which has characteristic $6$ and contains an element of order $6$. We now discuss how to construct such a ring in general. 
 
Let $m>1$ be an integer and let $\gamma$ be a formal variable. We denote by $$\cR_{m,r} = \Z_m[\gamma]/(\gamma^r-1)$$ the ring of univariate polynomials $\Z_m[\gamma]$ in $\gamma$ with coefficients in $\Z_m$ modulo the identity $\gamma^r = 1$.\footnote{The rings $\cR_{m,r}$ are sometimes denoted by $\Z_m[C_r]$ and referred to as the {\em group ring} of the cyclic group $C_r$ with coefficients in $\Z_m$. See e.g., \cite{KKS13,HH11} for some recent applications of these rings in cryptography.} More formally, each element $f \in \cR_{m,r}$ is represented by a degree $\leq r-1$ polynomial $f(\gamma) = \sum_{\ell=0}^{r-1}c_\ell \gamma^\ell$ with coefficients $c_i \in \Z_m$. Addition is done as in $\Z_m[\gamma]$ (coordinate wise modulo $m$) and multiplication is done over $\Z_m[\gamma]$ but using the identity $\gamma^r=1$ to reduce higher order monomials to degree $\leq r-1$. It is easy to see that this reduction is uniquely defined: to obtain the coefficient of $\gamma^\ell$ we sum all the coefficients of powers of $\gamma$ that are of the form $\ell + km$ for some integer $k \geq 0$. This implies the following lemma.

\begin{lem}
\label{nonzerolemma}
Let $f = \sum_{\ell=0}^{r-1}c_\ell \gamma^\ell$ be an element in $\cR_{m,r}$. Then, $f=0$ in the ring $\cR_{m,r}$ iff $c_i=0$ (in $\Z_m$) for all $0 \leq i \leq r-1$.
\end{lem}

\begin{remark}\label{nonzerodivisorgammapower}
For any $t\in\set{0,1,\cdots,r-1}$, $\gamma^t$ is not a zero divisor of the ring $\cR_{m,r}$. This holds since the coefficients of $\gamma^t \cdot f(\gamma)$ are the same as those of $f(\gamma)$ (shifted cyclicly $t$ positions).
\end{remark}


\subsection{Matrices over Commutative Rings}
Let $\cR$ be a commutative ring (with unity). Let $M\in \cR^{n\times n}$ be an $n\times n$ matrix with entries from $\cR$. Most of the classical theory of determinants can be derived in this setting in exactly the same way as over fields. One particularly useful piece of this theory is the Adjugate (or Classical Adjoint) matrix. For an $n \times n$ matrix $M \in \cR^{n \times n}$ the Adjugate matrix is denoted by $\adj(M) \in \cR^{n \times n}$ and has the $(j,i)$ cofactor of $A$ as its $(i,j)$th entry (recall that the cofactor is the determinant of the matrix obtained from $M$ after removing the $i$th row and $j$th column multiplied by $(-1)^{i+j}$). A basic fact in matrix theory is the following identity.
\begin{lem}[Theorem 1.7 from \cite{McDonald}]
\label{adjointlemma}
Let $M \in \cR^{n \times n}$ with $\cR$ a commutative ring with identity. Then
$M\cdot \adj(M)=\adj(M)\cdot M=\det(M)\cdot I_n$ where $I_n$ is the $n\times n$ identity matrix.
\end{lem}

The way we will use this fact is as follows: 

\begin{remark}\label{remark-adjugate}
Suppose $M\in \cR^{n \times n}$ has non-zero determinant and let $\ba = (a_1,\ldots,a_n)^t \in \cR^n$ be some column vector where $a_1=0$ or $a_1=c$, where $c$ is  not a zero-divisor. Then we can determine the value of $a_1$ (i.e., tell whether its $0$ or $c$) from the product $M \cdot \ba $. The way to do it is to multiply $M \cdot \ba$ from the left by $\adj(M)$ and to look at the first entry. This will give us $\det(M) \cdot a_1$ which is zero iff $a_1$ is (since $\det(M) \cdot c$ is always nonzero).
\end{remark}

\subsection{Matching Vector Families}
\begin{define}[Matching Vector Family]
Let $S\subset\Z_m\setminus\set{0}$ and let $\cF=(\cU,\cV)$ where $\cU=(\vect{\bu}{n}),\cV=(\vect{\bv}{n})$ and $\forall i\ \bu_i,\bv_i \in \Z_m^k$. Then $\cF$ is called an $S$-matching vector family of size $n$ and dimension $k$ if $\forall\ i,j$,
\begin{align*}
\inpro{\bu_i}{\bv_j}\begin{cases}
= 0 & \mbox{if } i=j\\
 \in S & \mbox{if } i\ne j
 \end{cases}
\end{align*}
If $S$ is omitted, it implies that $S=\Z_m\setminus\set{0}$.
\end{define}

\begin{thm}[Theorem 1.4 in \cite{Grolmusz99}]
\label{Grolmusz}
Let $m=p_1p_2\cdots p_r$ where $p_1,p_2\cdots, p_r$ are distinct primes with $r\ge 2$, then there exists an explicitly constructible $S$-matching vector family $\cF$  in $\Z_m^k$ of size $n\ge \bigexp{\Omega\left(\frac{(\log k)^r}{(\log\log k)^{r-1}}\right)}$ where $S=\set{a\in \Z_m: a\mod p_i \in \set{0,1}\  \forall\ i \in[r]}\setminus \set{0}$. 
\end{thm}
\begin{remark}
\label{CRT}
The size of $S$ in the above theorem is $2^r-1$ by the Chinese Remainder Theorem. Thus, there are matching vector families of size super-polynomial in the dimension of the space with inner products restricted to a set of size $2^r = |S \cup \{0\}|$. 
\end{remark}
In the special case when $p_1=2,p_2=3$, we have $m=6$ and the following corollary:
\begin{cor}
\label{Grolmuszmod6}
There is an explicitly constructible $S$-matching vector family $\cF$ in $\Z_6^k$ of size $n\ge \bigexp{\Omega\left(\frac{(\log k)^2}{\log\log k}\right)}$ where $S=\set{1,3,4}\subset \Z_6$
\end{cor}

\subsection{A number theoretic lemma}

We will need the following simple lemma. Recall that the {\em order} of an element $a$ in a finite multiplicative group  $G$ is the smallest integer $w \geq 1$ so that $a^w=1$.

\begin{lem}\label{lem-order}
Let $\F_p$ be a field of prime order $p$ and let $k \geq 1$ be an integer co-prime to $p$. Then, the algebraic closure of $\F_p$ contains an element $\zeta$ of order $k$.
\end{lem}
\begin{proof}
Since $k,p$ are co-prime, $p\in \Z_k^*$ which is the multiplicative group of invertible elements in $\Z_k$. Let $w \geq 1$ be the order of $p$ in the group $\Z_k^{*}$, so $k$ divides $p^w-1$. Consider the extension field $\F_{p^w}$, which is a sub field of the algebraic closure of $\F_p$. The multiplicative group $\F_{p^w}^*$ of this field is a cyclic group of size $p^w - 1$. Since $k$ divides this size, there must be an element in $\F_{p^w}$ of order $k$. 
\end{proof}

\section{Review of $O(n^{1/3})$ cost 2-server PIR}
\label{oldconstruction}
There are several known constructions of 2-server PIR with $O(n^{1/3})$ communication cost. We will recall here in detail a particular construction due to \cite{WoodruffY05} which uses polynomial interpolation using derivatives (over a field). In the next section we will replace the field with a ring and see how to use matching vector families to reduce the communication cost.

Let $\ba=(\vect{a}{n})$ be the database, choose $k$ to be smallest integer such that $n \le \binom{k}{3} $. Let $\F_q$ be a finite field with $q > 3$ elements and suppose for simplicity that $q$ is prime (so that partial derivatives behave nicely for polynomials of degree at most $3$). Let $\phi:[n]\mapsto \set{0,1}^k\subset \F_q^k$ be an embedding of the $n$ coordinates into points in $\set{0,1}^k$ of Hamming weight 3. Such an embedding exists since $n\le \binom{k}{3}$.

Define $F(\vect{x}{k})=F(\bx) \in\F_q[x_1,\cdots,x_k]$ as \[F(\bx)=\sum_{i=1}^n a_i\left( \prod_{j: \phi(i)_j=1}x_j\right)\] Note that $F(\bx)$ is a degree 3 polynomial satisfying $F(\phi(i))=a_i\ \forall\ i\in [n]$. Fix any two nonzero field elements $t_1 \neq t_2 \in \F_q \setminus \{0\}$. 

Suppose the user $\cU$ wants to recover the bit $a_\tau$. The protocol is as follows: The user picks a uniformly random element $\bz\in \F_q^k$ and sends $\phi(\tau)+t_1\bz$ to $\cS_1$ and $\phi(\tau)+t_2\bz$ to $\cS_2$. Each server $S_i$ then replies with the value of $F$ at the point received $F(\phi(\tau)+t_i\bz)$ as well as the values of the $k$ partial derivatives of $F$ at the same point $$\nabla F(\phi(\tau)+t_i\bz)=\left(\frac{\partial F}{\partial z_1}(\phi(\tau)+t_i\bz),\cdots,\frac{\partial F}{\partial z_k}(\phi(\tau)+t_i\bz)\right)$$ The partial derivatives here are defined in the same way as for polynomials over the real numbers.

\begin{empheq}[box=\fbox]{align*}
\cU &: \mathrm{Picks\ a\ uniformly\ random\ } \bz\in \F_q^k\\
\cU \rightarrow \cS_i &: \phi(\tau)+t_i\bz\\\
\cS_i \rightarrow \cU &: F( \phi(\tau)+t_i\bz), \nabla F( \phi(\tau)+t_i\bz)
\end{empheq}

The protocol is private since $\phi(\tau)+t\bz$ is uniformly distributed in $\F_q^k$ for any $\tau$ and $t\ne 0$. 

Consider the univariate polynomial $$g(t)=F(\phi(\tau)+t\bz).$$ Observe that, be the chain rule, $$g'(t)=\inpro{\nabla F(\phi(\tau)+t\bz)}{\bz}.$$ Thus the user can recover the values $g(t),g'(t)$ for $t=t_1,t_2$ from the server's responses. From this information the user needs to find $g(0)=F(\phi(\tau))=a_\tau$. Since $F$ is a degree 3 polynomial, $g(t)$ is a univariate degree 3 polynomial, let $g(t)=\sum_{\ell=0}^3 c_\ell t^\ell$. Therefore we have the following matrix equation:

\begin{align*}
\left[
\begin{matrix}
g(t_1)\\
g'(t_1)\\
g(t_2)\\
g'(t_2)\\
\end{matrix}
\right]
=
\left[
\begin{matrix}
1&t_1&t_1^2&t_1^3\\
0&1&2t_1&3t_1^2\\
1&t_2&t_2^2&t_2^3\\
0&1&2t_2&3t_2^2\\
\end{matrix}
\right]
\left[
\begin{matrix}
c_0\\
c_1\\
c_2\\
c_3
\end{matrix}
\right]
=M
\left[
\begin{matrix}
c_0\\
c_1\\
c_2\\
c_3
\end{matrix}
\right]
\end{align*}	

The matrix $M$ has determinant $det(M)=(t_2-t_1)^4$ and so M is invertible as long as $t_1 \neq t_2$. Thus the user can  find $c_0=g(0)=F(\phi(\tau)) = a_\tau$ by multiplying by the inverse of $M$.

The communication cost of this protocol is $O(k) = O(n^{1/3})$ since the user sends a vector in $\F_q^k$ to each server and each server sends an element in $\F_q$ and a vector in $\F_q^k$ to the user.

\section{The new 2-server scheme}
\label{newconstruction}

In this section we describe our main construction which proves Theorem \ref{mainthm}. Before describing the construction we set up some of the required ingredients and notations.

 The first ingredient is a matching vector family over $\Z_6$ as in Corollary~\ref{Grolmuszmod6}. That is, we construct an $S = \{1,3,4\}$- matching vector family $\cal F = (\cU,\cV)$ where $\cU=(\vect{\bu}{n}),\cV=(\vect{\bv}{n})$ have elements in $\Z_6^k$. Corollary~\ref{Grolmuszmod6} tells us that this can be done with $n=\exp(\Omega(\log^2 k/\log\log k))$ or  $k=\exp(\bigo{\sqrt{\log n\ \log\log n}})$.

We will work with polynomials over the ring $$\cR = \cR_{6,6}=\Z_6[\gamma]/(\gamma^6-1)$$ (see Section~\ref{preliminaries}).  We will denote the vector $(\gamma^{z_1},\gamma^{z_2},\cdots,\gamma^{z_k})$ by $\gamma^\bz$ where $\bz=(\vect{z}{k}) \in \Z_6^k$. We will need to extend the notion of partial derivatives to polynomials in $\cR[x_1,\ldots,x_k]$. This will be a non standard definition, but it will satisfy all the properties we will need. Instead of defining each partial derivative separately, we define one operator that will include all of them.

\begin{define}
	Let $\cR$ be a commutative ring and let $F(\bx)=\sum c_{\bz}\bx^{\bz} \in \cR[x_1,\ldots,x_k]$. We define $F^{(1)} \in (\cR^k)[x_1,\ldots,x_k]$ to be
	\begin{align*}
	F^{(1)}(\bx)&:=\sum (c_{\bz}\cdot \bz) \bx^{\bz}
	\end{align*}
\end{define}

For example, when $F(x_1,x_2)=x_1^2x_2+4x_1x_2+3x_2^2$ (with integer coefficients),
\[F^{(1)}(x_1,x_2)=\twovect{2}{1}x_1^2x_2+\twovect{4}{4}x_1x_2+\twovect{0}{6}x_2^2\]
One can think of $F^{(1)}$ both as a polynomial with coefficients in $\cR^k$ as well as a $k$-tuple of polynomials in $\cR[x_1,\ldots,x_k]$. This will not matter much since the only operation we will perform on $F^{(1)}$ is to evaluate it at a point in $\cR^k$.

\paragraph{The Protocol:}

Let $\ba=(a_1,a_2\cdots,a_n)\in \BF^n$ be an n-bit database shared by two servers $\cS_1$ and $\cS_2$. The user $\cU$ wants to find the bit $a_\tau$ without revealing any information about $\tau$ to either server. For the rest of this section, $\cR = \cR_{6,6} = \Z_6[\gamma]/(\gamma^6-1)$. The servers represent the database as a polynomial $F(\bx)\in \mathcal{R}[\bx]=\mathcal{R}[x_1,\cdots,x_k]$  given by
\[F(\bx)=F(x_1,\cdots,x_k)=\sum_{i=1}^n a_i \bx^{\bu_i},\] where $\cU = (\bu_1,\ldots,\bu_n)$ are given by the matching vector family $\cF = (\cU,\cV)$.

The user samples a uniformly random $\bz\in \Z_6^k$ and  then sends  $\bz+t_1\bv_\tau$ to $\cS_1$ and $\bz+t_2\bv_\tau$ to $\cS_2$ where we fix $t_1=0$ and $t_2=1$ (other choices of values would also work). $\cS_i$ then responds with the value of $F$ at the  point $\bgam^{\bz+t_i\bv_\tau}$, that is with $F(\bgam^{\bz+t_i\bv_\tau})$ and the value of the `first order derivative' at the same point $F^{(1)}(\bgam^{\bz+t_i\bv_\tau})$. Notice that the protocol is private since $\bz+t\bv_\tau$ is uniformly distributed over $\Z_6^k$ for any fixed $\tau$ and $t$.

\begin{empheq}[box=\fbox]{align*}
\cU &: \mathrm{Picks\ a\ uniformly\ random\ } \bz\in \Z_6^k\\
\cU\rightarrow \cS_i &: \bz+t_i\bv_\tau\\
\cS_i \rightarrow \cU &: F(\bgam^{\bz+t_i\bv_\tau}), F^{(1)}(\bgam^{\bz+t_i\bv_\tau})
\end{empheq}

\paragraph{Recovery:}
Define 
\[G(t):=F(\bgam^{\bz+t\bv_\tau})
=\sum_{i=1}^n a_i \gamma^{\inpro{\bz}{\bu_i}+t\inpro{\bv_\tau}{\bu_i}}\]
Using the fact that $\gamma^6=1$, we can rewrite $G(t)$ as:
\[G(t)=\sum_{\ell=0}^{5} c_\ell \cdot \gamma^{t\ell},\]
with each $c_\ell \in \cR$ given by $$c_\ell=\sum_{i:\inpro{\bu_i}{\bv_\tau}
=\ell\mod 6}a_i \gamma^{\inpro{\bz}{\bu_i}}.$$ Since
\[\inpro{\bu_i}{\bv_\tau}\mod 6 \,\,  \begin{cases}
= 0   &\mbox{if}\ i=\tau\\
\in  S=\set{1,3,4} &\mbox{if}\ i\ne \tau
\end{cases}
\]
we can conclude that $c_0=a_\tau\gamma^{\inpro{\bu_\tau}{\bz}}$ and $c_2=c_5=0$. Therefore $$G(t)=c_0+c_1\gamma^t+c_3\gamma^{3t}+c_4\gamma^{4t}.$$
Next, consider the polynomial $$g(T) = c_0+c_1T+c_3T^3+c_4T^4\in \mathcal{R}[T].$$ 
By definition we have
\begin{align*}
 g(\gamma^t)=G(t)=F(\bgam^{\bz+t\bv_\tau})\\
 g^{(1)}(\gamma^t)=\sum_{\ell=0}^5 \ell c_\ell \gamma^{t\ell}=\inpro{F^{(1)}(\bgam^{\bz+t\bv_\tau})}{\bv_\tau},
\end{align*}
where the last equality holds since $c_2=c_5=0$ and 
\begin{align*}
\inpro{F^{(1)}(\bgam^{\bz+t\bv_\tau})}{\bv_\tau}&=\left\langle\sum_{i=1}^n a_i \bu_i\gamma^{\inpro{\bz}{\bu_i}+t\inpro{\bv_\tau}{\bu_i}},\bv_\tau\right\rangle\\
&= \sum_{i=1}^n a_i\inpro{\bu_i}{\bv_\tau} \gamma^{\inpro{\bz}{\bu_i}+t\inpro{\bv_\tau}{\bu_i}}\\
&= \sum_{\ell=0}^{5} \ell \left(\sum_{i:\inpro{\bu_i}{\bv_\tau}
=\ell \mod 6}a_i \gamma^{\inpro{\bz}{\bu_i}}\right)\gamma^{t\ell}=\sum_{\ell=0}^5 \ell c_\ell \gamma^{t\ell}
\end{align*}
So the user can find the values of $g(\gamma^t),g^{(1)}(\gamma^t)$ for $t=t_1,t_2$. Since $t_1=0,t_2=1$, we obtain the following matrix equation:

\begin{align*}
\left[
\begin{matrix}
g(1)\\
g^{(1)}(1)\\
g(\gamma)\\
g^{(1)}(\gamma)\\
\end{matrix}
\right]
=
\left[
\begin{matrix}
1&1&1&1\\
0&1&3&4\\
1 &\gamma  &\gamma^3& \gamma^4 \\
0 &\gamma &3\gamma^3 &4\gamma^4
\end{matrix}
\right]
\left[
\begin{matrix}
c_0\\
c_1\\
c_3\\
c_4
\end{matrix}
\right]
=M
\left[
\begin{matrix}
c_0\\
c_1\\
c_3\\
c_4
\end{matrix}
\right]
\end{align*}	
The determinant (over $\cR$) of the matrix $M$ is
\begin{equation}
\label{determinant}
\det(M)=\gamma(\gamma-1)^4(\gamma^2+4\gamma+1)=3\gamma^5+4\gamma^4+3\gamma^3+2\gamma
\end{equation}
 and so, by Lemma~\ref{nonzerolemma}, is a non-zero element of the ring $\cR$. Since  $c_0=a_\tau\gamma^{\inpro{\bu_\tau}{\bz}}$, either $c_0=0$ or $c_0=\gamma^{\inpro{\bu_\tau}{\bz}}$ which is not a zero-divisor by remark \ref{nonzerodivisorgammapower}. Hence, by Remark~\ref{remark-adjugate}, the user can find whether $c_0=0$ from the vector $[g(1),g^{(1)}(1),g(\gamma),g^{(1)}(\gamma)]^t$ by multiplying it from the left by $\adj(M)$. Since $c_0=a_\tau\gamma^{\inpro{\bu_\tau}{\bz}}$, $a_\tau$ will be zero iff $c_0$ is and so the user can recover $a_\tau \in \{0,1\}$.


\paragraph{Communication Cost:}
The user sends a vector in $\Z_6^k$ to each server. Each server sends a element of $\cR$ and a vector in $\cR^k$ to the user. Since elements of $\cR$ have constant size description, the total communication cost is $O(k)=n^{o(1)}$.

\subsection{Working over $\Z_6$ or $\F_3$}\label{sec-overZm}
%

Using the  ring $\cR_{6,6}=\Z_6[\gamma]/(\gamma^6 - 1)$ in the above construction makes the presentation  clearer but is not absolutely necessary. Observing the proof, we see that one can replace it with any ring $\cR$ as long as  there is a homomorphism from $\cR_{6,6}$ to $\cR$ such that the determinant of the matrix $M$ (Eq. \ref{determinant}) doesn't vanish under this homomorphism. 

For example, we can work over the ring $\Z_6$ and use the element $-1$ as a substitute for $\gamma$. Since $(-1)^6 = 1$ all of the calculations we did with $\gamma$ carry through. In addition, the resulting determinant of $M$ is non zero when setting $\gamma= -1$ and so we can complete the recovery process. More formally, define the homomorphism $\tau:\Z_6[\gamma]/(\gamma^6 - 1)\mapsto \Z_6$ by extending the identity homomorphsim on $\Z_6$ using $\tau(\gamma)=-1$. Observe that the determinant of the matrix $M$ in Eq. (\ref{determinant}) doesn't vanish under this homomorphism, $\tau(\det(M))=-4=2$. 

A more interesting example is the ring of integers modulo $3$, which we denote by $\F_3$ to highlight that it is also a field. We can  use the homomorphsim $\phi: \Z_6[\gamma]/(\gamma^6-1)\mapsto \F_3$ by extending the natural homomorphsim from $\Z_6$ to $\F_3$ (given by reducing each element modulo $3$) using $\phi(\gamma)=-1$. Again the determinant in Eq. (\ref{determinant}) doesn't vanish. This also shows that our scheme can be made to be 	{\em bilinear}, as defined in \cite{RazborovY06}, since the answers of each server become linear combinations of database entries over a field.

\subsection{An Alternative Construction}

In the construction above we used the special properties of Grolmusz's construction, namely that the non-zero inner products are in the special set $S = \set{1,3,4}$. Here we show how to make the construction work with any matching vector family (over $\Z_6$). This construction also introduces higher order differential operators, which could be of use if one is to generalize this work further.

Suppose we run our protocol (with $\cR = \cR_{6,6}$) using a matching vector family with $S=\Z_6\setminus\set{0}$. Then, we cannot claim that $c_2=c_5=0$, but we still have $c_0=a_\tau\gamma^{\inpro{\bu_\tau}{\bz}}$. We can proceed by asking for the `second order' derivative of $F(\bx)=\sum_{i=0}^n a_i\bx^{\bu_i}$ which we define as
\[F^{(2)}(\bx):=\sum c_{\bz}\ (\bz\otimes\bz)\ \bx^{\bz}\]
where $\bz\otimes\bz$ is the $k\times k$ matrix defined by $(\bz\otimes\bz)_{ij}=z_iz_j$. For example, when $P(x_1,x_2)=x_1^2x_2+4x_1x_2+3x_2^2$,
\[P^{(2)}(x_1,x_2)=\fourmatrix{4}{2}{2}{1}x_1^2x_2+4\fourmatrix{1}{1}{1}{1}x_1x_2+3\fourmatrix{0}{0}{0}{4}x_2^2.\]

The final protocol is:
\begin{empheq}[box=\fbox]{align*}
\cU &: \mathrm{Picks\ a\ uniformly\ random\ } \bz\in \Z_m^k\\
\cU\rightarrow \cS_i &: \bz+t_i\bv_\tau\\
\cS_i \rightarrow \cU &: F(\bgam^{\bz+t_i\bv_\tau}), F^{(1)}(\bgam^{\bz+t_i\bv_\tau}),F^{(2)}(\bgam^{\bz+t_i\bv_\tau})
\end{empheq}
Notice that privacy is maintained and the communication is $O(k^2) = n^{o(1)}$ as before.

For recovery, define $g(T)\in \cR[T]$ as before and notice that, in addition to the identities
\begin{align*}
& g(\gamma^t)=\sum_{\ell=0}^5 c_\ell \gamma^{t\ell}=F(\bgam^{\bz+t\bv_\tau})\\
 & g^{(1)}(\gamma^t)=\sum_{\ell=0}^5 \ell c_\ell \gamma^{t\ell}=\inpro{F^{(1)}(\bgam^{\bz+t\bv_\tau})}{\bv_\tau},
\end{align*}
we also get the second order derivative of $g$ from 
$$g^{(2)}(\gamma^t)=\sum_{\ell=0}^5 \ell^2 c_\ell \gamma^{t\ell}=\inpro{F^{(2)}(\bgam^{\bz+t\bv_\tau})}{\bv_\tau\otimes\bv_\tau},$$ where the inner product of matrices is taken entry-wise and using the identity $\inpro{ { \bf u} \otimes { \bf u}}{ \bv \otimes \bv} = \inpro{ { \bf u}}{\bv}^2$.

By choosing $t_1=0,t_2=1$, we have the following matrix equation:

\begin{align*}
\left[
\begin{matrix}
g(1)\\
g^{(1)}(1)\\
g^{(2)}(1)\\
g(\gamma)\\
g^{(1)}(\gamma)\\
g^{(2)}(\gamma)
\end{matrix}
\right]
=
\left[
\begin{matrix}
1&1&1&1&1&1\\
0&1&2&3&4&5\\
0&1&4&9&16&25\\
1 &\gamma &\gamma^2 &\gamma^3& \gamma^4 &\gamma^5\\
0& \gamma &2\gamma^2 &3\gamma^3 &4\gamma^4& 5\gamma^5\\
0& \gamma &4 \gamma^2 &9\gamma^3 &16\gamma^4 &25\gamma^5\\
\end{matrix}
\right]
\left[
\begin{matrix}
c_0\\
c_1\\
c_2\\
c_3\\
c_4\\
c_5
\end{matrix}
\right]
=M
\left[
\begin{matrix}
c_0\\
c_1\\
c_2\\
c_3\\
c_4\\
c_5
\end{matrix}
\right]
\end{align*}
$\det(M)=4\gamma^3(\gamma-1)^9=4+2\gamma^3\ne 0$ and so we can use recover $a_\tau$ as before.

\section{Generalization to more servers}\label{sec-kserver}

In this section we prove Theorem~\ref{THM-kserver}. As was mentioned in the introduction, we will allow the database symbols to belong to a slightly larger alphabet $\Z_m$.

Let $q=2^{r-1}$ denote the number of servers $\cS_1,\cdots,\cS_q$ for some $r\ge 2$. Let $m=p_1p_2\cdots p_r$ where $p_1,p_2,\cdots,p_r$ are distinct primes. By theorem \ref{Grolmusz}, there is an explicit $S$-matching vector family $\cF=(\cU,\cV)$ of size $n$ and dimension $k=n^{\bigo{(\log\log n/\log n)^{1-1/r}}}$ where $S=\set{a\in \Z_m: a\mod p_i \in \set{0,1}\  \forall\ i \in[r]}\setminus \set{0}$. By remark \ref{CRT}, $|S\cup \set{0}|=2^r=2q$.

\paragraph{The Protocol:}

We will work over the ring $\cR=\cR_{m,m}=\Z_m[\gamma]/(\gamma^m-1)$. The servers represent the database $\ba=(\vect{a}{n})\in \Z_m^n$ as a polynomial $F(\bx)\in \mathcal{R}[\bx]=\mathcal{R}[x_1,\cdots,x_k]$  given by
\[F(\bx)=F(x_1,\cdots,x_k)=\sum_{i=1}^n a_i \bx^{\bu_i},\] where $\cU = (\bu_1,\ldots,\bu_n)$ are given by the matching vector family $\cF = (\cU,\cV)$.

The user samples a uniformly random $\bz\in \Z_m^k$ and  then sends  $\bz+t_i\bv_\tau$ to $\cS_i$  for $i\in [q]$ where $t_i=i-1$. $\cS_i$ then responds with the value of $F$ at the  point $\bgam^{\bz+t_i\bv_\tau}$, that is with $F(\bgam^{\bz+t_i\bv_\tau})$ and the value of the `first order derivative' at the same point $F^{(1)}(\bgam^{\bz+t_i\bv_\tau})$. Notice that the protocol is private since $\bz+t\bv_\tau$ is uniformly distributed over $\Z_m^k$ for any fixed $\tau$ and $t$.

\begin{empheq}[box=\fbox]{align*}
\cU &: \mathrm{Picks\ a\ uniformly\ random\ } \bz\in \Z_m^k\\
\cU\rightarrow \cS_i &: \bz+t_i\bv_\tau\\
\cS_i \rightarrow \cU &: F(\bgam^{\bz+t_i\bv_\tau}), F^{(1)}(\bgam^{\bz+t_i\bv_\tau})
\end{empheq}

\paragraph{Recovery:}
Similarly to the 2-server analysis, we define 
\[G(t):=F(\bgam^{\bz+t\bv_\tau})
=\sum_{i=1}^n a_i \gamma^{\inpro{\bz}{\bu_i}+t\inpro{\bv_\tau}{\bu_i}}=c_0+\sum_{\ell\in S}c_\ell\gamma^{t\ell},\]
and
$$g(T) = c_0+\sum_{\ell\in S}c_\ell T^\ell \in \mathcal{R}[T],$$ 
so that  $c_0=a_\tau\gamma^{\inpro{\bu_\tau}{\bz}}$ and
\begin{align*}
 g(\gamma^t)=G(t)=F(\bgam^{\bz+t\bv_\tau})\\
 g^{(1)}(\gamma^t)=\sum_{\ell=0}^{m-1} \ell c_\ell \gamma^{t\ell}=\inpro{F^{(1)}(\bgam^{\bz+t\bv_\tau})}{\bv_\tau},
\end{align*}
Hence, the user can calculate the values of $g(\gamma^t),g^{(1)}(\gamma^t)$ for $t=t_1,\cdots,t_q$ and we end up with the following (square) system of equations:

\begin{align*}
\left[
\begin{matrix}
g(\gamma^{t_1})\\
g^{(1)}(\gamma^{t_1})\\
\vdots\\
g(\gamma^{t_q})\\
g^{(1)}(\gamma^{t_q})\\
\end{matrix}
\right]
=
\left[
\begin{matrix}
1&\cdots&\gamma^{t_1\ell}&\cdots\\
0&\cdots&\ell\gamma^{t_1\ell}&\cdots\\
\vdots& &\vdots& &\vdots\\
1&\cdots&\gamma^{t_q\ell}&\cdots\\
0&\cdots&\ell\gamma^{t_q\ell}&\cdots\\
\end{matrix}
\right]
\left[
\begin{matrix}
c_0\\
\vdots\\
c_\ell\\
\vdots
\end{matrix}
\right]
=M
\left[
\begin{matrix}
c_0\\
\vdots\\
c_\ell\\
\vdots
\end{matrix}
\right]
\end{align*}	
where the $2^r = 2q$ columns are indexed by $\ell\in \set{0}\cup S$. Instead of computing the determinant (and the Adjugate matrix), we will use the following Lemma (proven below).

\begin{lem}\label{lem-lambda}
There exists a row vector $$\blam=[\alpha_1,\beta_1, \cdots ,\alpha_q, \beta_q]\in \cR^{2q}$$ such that $\blam M=[\mu,0,\cdots,0]$ for some $\mu\in \cR$ where $\mu \mod p_i \ne 0\ \forall i\in[r]$.
\end{lem}

Using this Lemma, the user can recover  $a_\tau$ as follows. We have
\begin{align*}
\nu :=
\blam
\left[
\begin{matrix}
g(\gamma^{t_1})\\
g^{(1)}(\gamma^{t_1})\\
\vdots\\
g(\gamma^{t_q})\\
g^{(1)}(\gamma^{t_q})\\
\end{matrix}
\right]
=
\blam M
\left[
\begin{matrix}
c_0\\
\vdots\\
c_\ell\\
\vdots
\end{matrix}
\right]
=[\mu,0,\cdots,0]
\left[
\begin{matrix}
c_0\\
\vdots\\
c_\ell\\
\vdots
\end{matrix}
\right]
=\mu c_0
\end{align*}	
Taking this equation modulo $p_i$ we get, \[(\nu \mod p_i)=(\mu c_0 \mod p_i) = (\mu \mod p_i)(a_\tau \mod p_i) \gamma^{\inpro{\bu_\tau}{\bz}}\]
Let $\mu=\sum_{j=0}^{m-1}\mu_j\gamma^j$ and  $\nu=\sum_{j=0}^{m-1}\nu_j\gamma^j$. Since $\mu \mod p_i \ne 0$, there exists $j$ such that $\mu_j \mod p_i \ne 0$. So $(a_\tau\mod p_i)=(\mu_j \mod p_i)^{-1}(\nu_{j+\inpro{\bu_\tau}{\bz}} \mod p_i)$. So we can find $a_\tau \mod p_i$ for each $i\in[r]$. Finally we use Chinese Remainder Theorem to find $a_\tau \in \Z_m$.

\subsection{Proof of Lemma~\ref{lem-lambda}}
For any $\blam=[\alpha_1,\beta_1, \cdots ,\alpha_q, \beta_q]\in \cR^{2q}$ we can define a function $h:S\cup\set{0}\mapsto \cR$ as:
\[h(\ell) =(\blam M)_\ell = \left(\sum_{i=1}^q \alpha_i \gamma^{t_i\ell} \right)+ \ell \left(\sum_{i=1}^q \beta_i \gamma^{t_i \ell}\right).\]
Our goal is then to construct an $h$ of this form  such that 
\begin{align*}
h(\ell)
\begin{cases}
= 0 &\mbox{if}\ \ell\in S\\
= \mu & \mbox{if}\ \ell=0
\end{cases}
\end{align*}
where $(\mu \mod p_i) \ne 0\ \forall i\in[r]$.

Notice that, by Chinese Remaindering,
\begin{equation}\label{eq-isomorphism}
\cR = \cR_{m,m} \cong \cR_{p_1,m} \times \ldots \times \cR_{p_r,m},	
\end{equation} 
where we recall that $\cR_{p_i,m} = \Z_{p_i}[\gamma]/(\gamma^m-1)$. Therefore, we also get that, for a formal variable $x$, the rings of univariate polynomials also satisfy
$$ \cR[x] \cong \cR_{p_1,m}[x] \times \ldots \times \cR_{p_r,m}[x]. $$
In other words, any family of polynomials $f_i \in \cR_{p_i,m}[x]$, $i\in [r]$ can be `lifted' to a single polynomial $f \in \cR[x]$ so that $ (f \mod p_i) = f_i$ for all $i$ (reducing $f$ mod $p_i$ is done coordinate-wise). Moreover, since this lift is done coefficient-wise (using Eq.\ref{eq-isomorphism}), we get that the degree of $f$ is equal to the maximum of the degrees of the $f_i$'s. 

We begin by constructing, for each $i \in [r]$ the following polynomial $f_i(x)\in \cR_{p_i,m}[x]$:
\[f_i(x)=\prod_{\ell\in S,\ \ell=0\mod p_i}(x-\gamma^\ell)\]
The degree of $f_i$ is $2^{r-1}-1=q-1$ so, by the above comment, we can find a  polynomial $f(x)\in \cR[x]$ of degree $q-1$ such that $f(x)\equiv f_i(x) \mod p_i$ for all $i\in [r]$. Define $\alpha_i, i\in[q]$ to be the coefficients of the polynomial $f$ so that $f(x)=\sum_{i=1}^q \alpha_{i}x^{i-1}$. Since we defined $t_i=i-1$, we have $f(x)=\sum_{i=1}^q \alpha_{i}x^{t_i}$. Define $\beta_i=-\alpha_i$ for all $i\in [q]$. Our final construction of $h$ is thus
\[h(\ell)=f(\gamma^\ell)-\ell f(\gamma^\ell)\]

\begin{claim}
$h(\ell)=0\ \forall \ell\in S$
\end{claim}
\begin{proof}
Since $0\notin S$, $\ell\ne 0$. We will look at $h(\ell)$ modulo each of the primes.
\begin{align*}
h(\ell) \mod p_i = f_i(\gamma^\ell)-(\ell\mod p_i) f_i(\gamma^\ell)=
\begin{cases}
f_i(\gamma^\ell)= 0 & \mbox{if}\ \ell=0 \mod p_i\\
f_i(\gamma^\ell)-f_i(\gamma^\ell) =0 & \mbox{if}\ \ell=1 \mod p_i
\end{cases}
\end{align*}
Therefore, using Chinese Remaindering, $h(\ell)=0\ \forall \ell\in S$.
\end{proof}
\begin{claim}
 $(h(0) \mod p_j)\ne 0$ for all $j\in [r]$
\end{claim}
\begin{proof}
Suppose in contradiction that $(h(0) \mod p_j)= 0$, then
\[h(0) \mod p_j=f_j(1)=\prod_{\ell\in S,\ \ell=0\mod p_j}(1-\gamma^\ell)=0.\]
The above equation holds in the ring $\left(\Z_{p_j}[\gamma]/(\gamma^m-1)\right)$.Therefore, if we consider what happens in the ring $\Z_{p_i}[\gamma] \cong \F_{p_i}[x]$ (we replace the formal variable $\gamma$ with $x$ to highlight the fact that $x$ does not satisfy any relation) we get that
\begin{equation}\label{eq-identity}
	\prod_{\ell\in S,\ \ell=0\mod p_j}(1-x^\ell)=(x^m-1)\theta(x)
\end{equation}
for some polynomial $\theta(x)\in \F_{p_j}[x]$. The above equation is an identity in the ring $\F_{p_j}[x]$. So we can check its validity by substituting values for $x$ from the algebraic closure  of $\F_{p_j}$. Let $m' = m/p_j$ and let $\zeta$ be an element in the algebraic closure of $\F_{p_j}$ of order $m'$ (so $\zeta^\ell=1$ iff $m'$ divides $\ell$). Since $m'$ and $p_j$ are co-prime, such an element exists by Lemma~\ref{lem-order}. If we substitute $\zeta$ into Eq.~\ref{eq-identity}, the RHS is zero (since $m'$ divides $m$). However, each term in the LHS product is nonzero, since if $\ell =0 \mod p_j$ and $m'$ divides $\ell$ then $\ell = 0 \mod m$ but we know that $0\notin S$. Since we are working over the algebraic closure of $\F_{p_j}$ which is a field, the product of nonzero elements is nonzero. This is a contradiction, and so Eq.~\ref{eq-identity} does not hold.
\end{proof}

\section{Concluding remarks}\label{sec-conclude}
In this work we presented the first 2-server PIR scheme (information theoretic) with sub-polynomial cost. It is unclear what is the optimal communication cost of 2-server schemes and we conjecture that our protocol is far from optimal.

One approach to decrease the communication cost is to take $m$ to be a product of $r>2$ prime factors in theorem \ref{Grolmusz} to get a larger $S$-matching vector family where $S=\set{a\in \Z_m: a\mod p_i \in \set{0,1}\  \forall\ i \in[r]}\setminus \set{0}$ which is of size $2^{r}-1$. So we need $2^{r-1}$ independent equations from each server to find $c_0$. We can ask the servers for derivatives of $F$ at $\gamma^{\bz+t\bv_\tau}$ up to order $2^{r-1}-1$. If these equations are `independent' i.e. the determinant of the coefficient matrix doesn't vanish then we can find $c_0$. If we can do this, we can decrease the cost to $n^{\bigo{2^r(\log\log n/\log n)^{1-1/r}}}$. But observe that for each $l\in S$, $l^2=l\mod m$ since $l\mod p_i \in \set{0,1}\ \forall i\in[r]$. So higher order derivatives of $g$ are equal to the first order derivative and we get repeated rows in the coefficient matrix $M$. One avenue for improvement could be by trying to construct $S$ such that elements of $S$ doesn't satisfy a low-degree monic polynomial.

\section{Acknowledgements} We would like to thank Klim Efremenko and Sergey Yekhanin for helpful comments. 

\bibliographystyle{alpha}

\bibliography{bibliography}

\newcommand{\etalchar}[1]{$^{#1}$}
\begin{thebibliography}{CKGS98}

\bibitem[Amb97]{Ambainis97}
Andris Ambainis.
\newblock Upper bound on communication complexity of private information
  retrieval.
\newblock In {\em ICALP}, pages 401--407, 1997.

\bibitem[BI01]{BeimelI01}
Amos Beimel and Yuval Ishai.
\newblock Information-theoretic private information retrieval: A unified
  construction.
\newblock In {\em ICALP}, pages 912--926, 2001.

\bibitem[BIKR02]{BeimelIKR02}
Amos Beimel, Yuval Ishai, Eyal Kushilevitz, and Jean-Fran\c{c}ois Raymond.
\newblock Breaking the $o(n^{1/(2k-1)})$ barrier for information-theoretic
  private information retrieval.
\newblock In {\em FOCS}, pages 261--270, 2002.

\bibitem[CFL{\etalchar{+}}13]{CheeFLWZ13}
Yeow~Meng Chee, Tao Feng, San Ling, Huaxiong Wang, and Liang~Feng Zhang.
\newblock Query-efficient locally decodable codes of subexponential length.
\newblock {\em Computational Complexity}, 22(1):159--189, 2013.

\bibitem[CKGS98]{ChorKGS98}
Benny Chor, Eyal Kushilevitz, Oded Goldreich, and Madhu Sudan.
\newblock Private information retrieval.
\newblock {\em J. ACM}, 45(6):965--981, 1998.

\bibitem[DGY10]{DvirGY10}
Zeev Dvir, Parikshit Gopalan, and Sergey Yekhanin.
\newblock Matching vector codes.
\newblock In {\em FOCS}, pages 705--714, 2010.

\bibitem[Efr09]{Efremenko09}
Klim Efremenko.
\newblock 3-query locally decodable codes of subexponential length.
\newblock In {\em STOC}, pages 39--44, 2009.

\bibitem[Gar]{Gasarch}
William Gararch.
\newblock A webpage on private information retrieval.
\newblock \url{https://www.cs.umd.edu/~gasarch/TOPICS/pir/pir.html}.

\bibitem[Gas04]{Gasarch04}
William~I. Gasarch.
\newblock A survey on private information retrieval (column: Computational
  complexity).
\newblock {\em Bulletin of the EATCS}, 82:72--107, 2004.

\bibitem[Gro99]{Grolmusz99}
Vince Grolmusz.
\newblock Superpolynomial size set-systems with restricted intersections mod 6
  and explicit ramsey graphs.
\newblock {\em Combinatorica}, 20:2000, 1999.

\bibitem[HH11]{HH11}
Barry Hurley and Ted Hurley.
\newblock Group ring cryptography.
\newblock {\em CoRR}, abs/1104.1724, 2011.

\bibitem[IS10]{ItohS10}
Toshiya Itoh and Yasuhiro Suzuki.
\newblock Improved constructions for query-efficient locally decodable codes of
  subexponential length.
\newblock {\em IEICE Transactions}, 93-D(2):263--270, 2010.

\bibitem[KdW03]{KerenidisW03}
Iordanis Kerenidis and Ronald de~Wolf.
\newblock Exponential lower bound for 2-query locally decodable codes via a
  quantum argument.
\newblock In {\em STOC}, pages 106--115, 2003.

\bibitem[KS13]{KKS13}
C.~Koupparis Kahrobaei and V.~Shpilrain.
\newblock Public key exchange using matrices over group rings.
\newblock {\em Groups, Complexity, and Cryptology}, 5:97--115, 2013.

\bibitem[KT00]{KT00}
Jonathan Katz and Luca Trevisan.
\newblock On the efficiency of local decoding procedures for error-correcting
  codes.
\newblock In {\em 32nd ACM Symposium on Theory of Computing (STOC)}, pages
  80--86, 2000.

\bibitem[Lip]{Lipmaa}
Helger Lipmaa.
\newblock A webpage on oblivious transfer or private information retrieval.
\newblock
  \url{http://www.cs.ut.ee/~lipmaa/crypto/link/protocols/oblivious.php}.

\bibitem[McD84]{McDonald}
B.~R. McDonald.
\newblock {\em Linear Algebra Over Commutative Rings}.
\newblock Pure and Applied Mathematics \#87. Marcel Dekker, New York, 1984.

\bibitem[RY06]{RazborovY06}
Alexander~A. Razborov and Sergey Yekhanin.
\newblock An {$\Omega(n^{1/3})$} lower bound for bilinear group based private
  information retrieval.
\newblock In {\em FOCS}, pages 739--748, 2006.

\bibitem[WdW05]{WehnerW05}
Stephanie Wehner and Ronald de~Wolf.
\newblock Improved lower bounds for locally decodable codes and private
  information retrieval.
\newblock In {\em ICALP}, pages 1424--1436, 2005.

\bibitem[WY05]{WoodruffY05}
David~P. Woodruff and Sergey Yekhanin.
\newblock A geometric approach to information-theoretic private information
  retrieval.
\newblock In {\em IEEE Conference on Computational Complexity}, pages 275--284,
  2005.

\bibitem[Yek08]{Yekhanin08}
Sergey Yekhanin.
\newblock Towards 3-query locally decodable codes of subexponential length.
\newblock {\em J. ACM}, 55(1), 2008.

\bibitem[Yek12]{Yekhanin12}
Sergey Yekhanin.
\newblock Locally decodable codes.
\newblock {\em Foundations and Trends in Theoretical Computer Science},
  6(3):139--255, 2012.

\end{thebibliography}

\end{document}